\documentclass[conference]{IEEEtran}
\IEEEoverridecommandlockouts
\usepackage{amsmath,amsthm,amsfonts,dsfont,bm}
\allowdisplaybreaks
\usepackage{physics}
\usepackage{textgreek}
\usepackage{enumitem}
\usepackage{multirow}
\usepackage{algpseudocode}
\usepackage{algorithm}
\usepackage{xcolor}
\usepackage{soul}
\usepackage{tikz}
\usetikzlibrary{positioning}
\usepackage[noadjust]{cite}
\usepackage{comment}

\usepackage{subcaption}
\usepackage{wrapfig}

\definecolor{Mycolor}{rgb}{.2,.2,.8}

\newtheorem{theorem}{Theorem}

\newcommand{\Xc}{\mathcal{X}}

\newcommand{\Pb}{\mathbb{P}}

\newcommand{\Eb}{\mathbb{E}}
\newcommand{\Rb}{\mathbb{R}}
\newcommand{\Zb}{\mathbb{Z}}

\newcommand{\dm}{\mathrm{d}}

\usepackage{amssymb}

\title{\LARGE{Staggered Quantizers for Perfect Perceptual Quality: A Connection between Quantizers with Common Randomness and Without}}

\author{\IEEEauthorblockN{Ruida Zhou}
\IEEEauthorblockA{{Department of Electrical Engineering} \\
{University of California, Los Angeles}\\
Los Angeles, CA\\
ruida@g.ucla.edu}
\and
\IEEEauthorblockN{Chao Tian}
\IEEEauthorblockA{{Department of Electrical and Computer Engineering} \\
{Texas A\&M University}\\
College Station, TX \\
chao.tian@tamu.edu}
}

\begin{document}
\maketitle

\begin{abstract}
The rate-distortion-perception (RDP) framework has attracted significant recent attention due to its application in neural compression. It is important to understand the underlying mechanism connecting procedures with common randomness and those without. Different from previous efforts, we study this problem from a quantizer design perspective. By analyzing an idealized setting, we provide an interpretation of the advantage of dithered quantization in the RDP setting, which further allows us to make a conceptual connection between randomized (dithered) quantizers and quantizers without common randomness. This new understanding leads to a new procedure for RDP coding based on staggered quantizers. 
\end{abstract}

\section{Introduction} \label{sec:intro}

Compression plays an important role in the efficient representation of information content, particularly visual content. Traditionally, the tradeoff between the compression rate and the incurred distortion has been studied under two different but related frameworks: the quantization framework \cite{gersho1992vector} and the rate-distortion theory \cite{berger1971rate} framework. In the former, the focus is on the design of quantizers that compress data samples one at a time (i.e., scalar quantization) or a few at a time (i.e., vector quantization), while the latter focuses on the fundamental limits of lossy compression by allowing an asymptotically large number of samples to be encoded together. 


Largely driven by the recent emergence of the neural compression, the issue of perceptual quality has led to the study of the problem of rate-distortion-perception (RDP) tradeoff \cite{blau2019rethinking,wagner2022rate,theis2021advantages,chen2022rate,serra2023computation,hamdi2023rate,salehkalaibar2024rate,niu2023conditional}. In this formulation, a new quality constraint, which was introduced to capture the perceptual quality loss due to compression, is further imposed in addition to the existing objective distortion constraint. Mathematically, this formulation \cite{blau2019rethinking} requires the probability distribution of the content after decompression to be close to that of the source content before compression; the case when the two distributions are exactly the same is often referred to as ``perfect perceptual quality'', which is our focus in this work. 

The RDP problem has attracted significant recent research attention, and several studies in this area revealed that common randomness plays an important role in this setting \cite{theis2021advantages,chen2022rate}. More precisely, the lack of common randomness can cause significant performance loss compared to methods that have such common randomness at their disposal, and this loss is particularly severe for scalar quantization. There are two known prevailing methods of introducing common randomness for RDP coding. The first is based on probabilistic sampling \cite{li2018strong}, and the second is through universal dithered quantization  \cite{ziv1985universal,zamir1992universal}. The first approach requires the knowledge of a target joint distribution between the samples and the compressed version, and furthermore, involves a rather complex sampling procedure. 
The dither-based approach, on the other hand, is simpler to implement and thus more attractive, however, its architecture places an inherent constraint on the eventual probability distribution, and though widely used, it is not clear what actually makes it suitable for the RDP setting. 

One piece of the puzzle has thus far been missing between the compression procedures without common randomness (e.g., scalar quantization with deterministic encoder) and those with a large amount of common randomness (dithered quantizers), particularly from a quantizer design perspective. That is, quantizers with deterministic encoders require no common randomness, and the dither-based approach will utilize common randomness on an uncountable set in a less transparent manner. What exactly is the underlying mechanism that lends the dither-based approach the advantage, and is there an effective procedure with an intermediate amount of common randomness? Although these questions have previously been studied under the rate-distortion framework with asymptotic large sample block size \cite{saldi2014randomized}, the asymptotic nature of such analysis makes the mechanism rather opaque. 

In this work, we develop a better understanding of these issues under the quantization framework. Using a decomposition perspective, we provide a new way to understand the mechanism from which procedures utilizing common randomness obtain the advantage. We first focus on an idealized setting on the unit circle, and provide a complete analysis of the performance. Based on these understandings, we provide a new approach to introduce common randomness using staggered quantizers. We further discuss the application of such an approach to other sources. It should be noted that staggered quantizers have been previously used for multiple description coding \cite{tian2005new,tian2005staggered,samarawickrama2010m} which offered surprisingly competitive performance compared to more sophisticated approaches.


\section{Backgrounds}\label{sec:background}

\subsection{Rate-distortion function and quantizers} \label{sec:RD-Quan}

Let the data source $X$ be a real-valued random variable, with a distribution $P_X$ on the alphabet $\mathcal{X}$. The reconstruction alphabet is denoted as $\hat{\mathcal{X}}$. Given a distortion measure $d: \mathcal{X}\times \hat{\mathcal{X}}\rightarrow [0,\infty)$, e.g., the squared error distortion $d(x,\hat{x})=(x-\hat{x})^2$ when $\mathcal{X}=\hat{\mathcal{X}}=\mathbb{R}$, the (informational) rate-distortion function under a distortion constraint $D$ is defined as
\begin{align*}
R(D)=\min_{P_{\hat{X}|X}: \mathbb{E}{d(X,\hat{X})}\leq D}I(X;\hat{X}),
\end{align*}
where $I(\cdot;\cdot)$ is the mutual information function. 

Rate-distortion theory deals with the setting when an infinite number of samples is allowed to be encoded together. In practice, samples are usually encoded one or few at a time, referred to as scalar quantization and vector quantization, respectively. In particular, a scalar quantizer consists of an encoding mapping $f:\mathcal{X}\rightarrow \mathbb{Z}$ which determines the representation index to assign to a sample, and a decoding function $g: \mathbb{Z}\rightarrow \hat{\mathcal{X}}$ which assigns a reconstruction point to each representation index. Therefore, $\hat{X}=g(f(X))$. Indices are allowed to be further entropy-coded, e.g., using Huffman code. When entropy coding is allowed, it is usually referred to as entropy-constrained scalar quantization (ECSQ), whereas when the number of quantization levels is fixed, it is usually referred to as fixed-rate quantization. 

Universal dithered quantizer utilizes a uniform quantizer with stepsize $\Delta$ in the encoding and decoding process \cite{zamir2014lattice}. Different from classic deterministic quantizers, a random noise $Z$, independent of the data samples and uniformly distributed on the base interval $(-\Delta/2,\Delta/2]$, is available at both the encoder and the decoder. The noise $Z$ is first added to the sample as $X+Z$, which is then quantized to its nearest neighbor using the deterministic uniform quantizer, and finally the same dither noise $Z$ is subtracted at the decoder. It was shown \cite{ziv1985universal,zamir1992universal} that using this procedure $\hat{X}=X+\tilde{Z}$, where $\tilde{Z}$ has the same marginal probability distribution as $Z$ and is also independent of $X$, and conditioned on the common randomness, the optimal entropy coding rate (of the lattice index) is exactly $H(f(X+Z)|Z)=I(X;X+Z)$. 
Note that such a rate is impossible to achieve in practice, since it requires one entropy code for a specific realization of the noise $Z=z$: Firstly, the usual technique of universal compression becomes unrealistic because it is unlikely (with zero probability) to have identical noise realizations and therefore very few samples to estimate the corresponding probability distribution; secondly, unless the distribution is analytically simple, storing the distribution or the entropy coding codewords for each noise realization is also unrealistic. Entropy coding of $f(X + Z)$ can be considered instead, resulting in a rate of $H(f(X+Z))$.

\subsection{Rate-distortion-perception function and RDP coding} \label{sec:RDP-coding}

The (informational) rate-distortion-perception function can be viewed as a generalization of the rate-distortion function, which under a given distortion constraint $D$ and a given perception constraint $P$, is defined as
\begin{align}
R(D,P)=\min_{P_{\hat{X}|X}: \mathbb{E}{d(X,\hat{X})}\leq D, w(P_X,P_{\hat{X}})\leq P}I(X;\hat{X}), \label{eqn:RDP}
\end{align}
where $w(\cdot,\cdot)$ is a measure quantifying the distance between two probability distributions, e.g., KL divergence, total variation, or Wasserstein metric. We are mainly interested in the case of perfect perception, i.e., 
\begin{align}
R(D,0)=\min_{P_{\hat{X}|X}: \mathbb{E}{d(X,\hat{X})}\leq D, P_{\hat{X}}=P_{X}}I(X;\hat{X}), \label{eqn:RDP0}
\end{align}
which is independent of the choice of $w(\cdot, \cdot)$ measure. Similar to the rate-distortion setting, it was shown \cite{theis2021coding} that the RDP function is also the fundamental limit of any encoding and decoding function pairs in the RDP setting. It was established in \cite{yan2021perceptual} that under the MSE distortion measure, $R(D,0)= R(\frac{D}{2},\infty).$
These results are again asymptotic in nature, meaning the corresponding codes are allowed to encode a large number of samples together. 

For scalar quantization (also called one-shot coding), it is possible to achieve the following coding rate \cite{theis2021coding}
$R(D,P)+\log(R(D,P)+1)+4$, using the sampling-based approach mentioned earlier, which is at a higher rate than the RDP function. The loss can be significant at the usual range of practical compression applications, e.g., at a target rate of $4$bits with a potential loss of more than $4$bits. It is not known whether this is the best rate possible for one-shot coding. 

It has been shown that quantizers without common randomness can suffer significantly in RDP coding, and common randomness is important. Dithered quantizer appears to be a natural match and can be utilized. However, the output of the original dithered quantizer has a distribution the same as $X+\tilde{Z}$, and therefore, there is a mismatch with the target RDP-optimal distribution. Particularly, for the perfect perceptual quality setting, the distribution of $X+\tilde{Z}$ may be different from $P_X$, and a distribution shaping procedure is needed at the decoder, at the expense of increased distortion. This shaping can be accomplished using a nonlinear function $\phi(\cdot)$ operating on the output of the dithered quantizer $X + \tilde{Z}$, and neural networks can be used to fulfill this role. 

\subsection{Quantization on the unit circle}

Consider the following idealized \textit{unit-circle} setting: the data signal $X$ to be compressed is uniformly distributed over the unit circle $\Xc = \{x \in \Rb^2: \|x\|_2 = 1\}$. The distortion is measured using the square error function  $d(x,\hat{x})=\| x - \hat{x} \|_2^2$, the coding rate is set at $1$ bit per sample, and the reconstruction $\hat{X}$ is required to be of perfect perception quality, i.e., $\hat{X} \stackrel{d}{=} X$. Since the signal has its domain is the unit circle, we can represent any $x \in \Xc$ by its angle $\theta(x) \in \Theta \triangleq (-\pi, \pi]$ such that $x = (\cos(\theta(x)), \sin(\theta(x)))$. Fixed-rate quantization at rate 1 on this data source was previously considered in \cite{theis2021advantages} to illustrate the advantage of stochastic (dithered) encoders. Two types of quantizers were considered  in \cite{theis2021advantages}: 

\begin{itemize}

\item {Quantizer with a deterministic encoder}: Since there is no common randomness, to obtain perfect perception quality, decoder side noise must be injected. It was shown that the optimal quantization procedure in this case is as follows:
\begin{align*}
f(\theta(x))=\left\{
\begin{matrix}
1&  \theta(x)\in [0, \pi)\\
-1& \text{otherwise}
\end{matrix}
\right., \quad g(i) = \frac{i\times \pi}{2}- \bar{Z},
\end{align*}
where $\bar{Z}$ is a private random variable at the decoder side, independent of $X$, distributed uniformly on $[-\pi/2, \pi/2)$. We here view $g(i)$ as a random function, and therefore did not include $\bar{Z}$ as part of the function input. This procedure gives a distortion $2-8/\pi^2$.

\item {Dithered quantizer}: Let $Z$ be distributed uniformly over $[-\pi/2, \pi/2)$ independent of $X$, dithered quantization operates as follows:
\begin{align*}
f(Y)=\left\{
\begin{matrix}
1& Y  \in [0, \pi) \textnormal{ mod } 2\pi \\
-1& \text{otherwise}
\end{matrix}
\right.,~ g(i) = \frac{i\times \pi}{2} - Z,
\end{align*}
where $Y = \theta(x)+Z$ and $\theta(\hat{x})=g(f(\theta(x)+Z))$. By the property of the dither quantizer, we have $\theta(\hat{X}) = \theta(X) + \tilde{Z} \mod 2\pi$, where $\tilde{Z} \stackrel{d}{=} Z$ and is independent of $X$. The distortion thus induced is $2-4/\pi$, which is about 38.9\% lower than that using the deterministic encoder.
\end{itemize}

The dithered quantizer performs better here for two reasons: 1) The distribution of $\theta(X) + \tilde{Z} \mod 2\pi$ is uniform on the unit circle, and thus naturally matches the perceptual requirement; 2) If the perception consideration were not present, the first approach could choose a single reconstruction point to minimize the distortion, however now it is forced to utilize private randomness at the decoder, over $1/2$ of the unit circle, to produce the desired distribution; this private randomness thus induces additional distortion. Fig. \ref{fig:illustration} (a) and (b) illustrate this effect of the two procedures.

\section{Quantization on the unit circle}
\label{sec:unitcircle}

\usetikzlibrary{shapes.misc}
\tikzset{cross/.style={cross out, draw=black, fill=none, minimum size=2*(#1-\pgflinewidth), inner sep=1pt, outer sep=1pt}, cross/.default={3pt}}

\newcommand\Iteration{3}
\newcommand\radius{1.5}
\newcommand\coor{\radius*5/4}
\def\SignalPos{35}

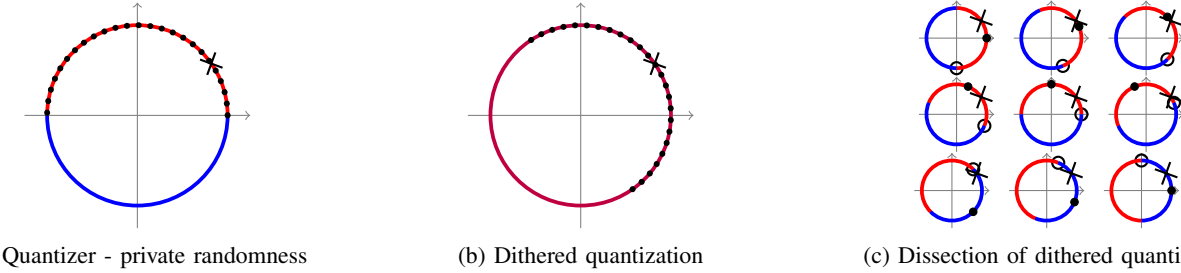
\begin{figure*}[ht]

\centering
\def\radius{1.2}
\begin{subfigure}[b]{0.31\textwidth}
\centering
\begin{tikzpicture}
\draw[] (0,0) circle (\radius cm);
\draw[gray, ->] (-\coor,0) -- (\coor,0);
\draw[gray, ->] (0,-\coor) -- (0,\coor);
 \draw[line width=0.5mm, red] ({\radius*cos(0)}, {\radius*sin(0)}) arc ({0:180:\radius});
\draw[line width=0.5mm, blue] ({\radius*cos(0)}, {\radius*sin(0)}) arc ({0:-180:\radius});
\draw[thick] ({\radius*cos(\SignalPos)}, {\radius*sin(\SignalPos)})  node[cross,rotate=25] {};
\draw[line width=0.8mm, line cap=round, dash pattern=on 0mm off 1.9\pgflinewidth] 
({\radius*cos(0)}, {\radius*sin(0)}) arc ({0:180:\radius});
\end{tikzpicture}
\caption{Quantizer - private randomness}
\label{fig:single}
\end{subfigure}
\hfill
     \begin{subfigure}[b]{0.28\textwidth}
         \centering
\begin{tikzpicture}
\draw[] (0,0) circle (\radius cm);
\draw[gray, ->] (-\coor,0) -- (\coor,0);
\draw[gray, ->] (0,-\coor) -- (0,\coor);
 \draw[line width=0.5mm, purple] ({\radius*cos(0)}, {\radius*sin(0)}) arc ({0:360:\radius});
\draw[thick] ({\radius*cos(\SignalPos)}, {\radius*sin(\SignalPos)})  node[cross,rotate=25] {};
\draw[line width=0.8mm, line cap=round, dash pattern=on 0mm off 1.9\pgflinewidth] ({\radius*cos(\SignalPos-90)}, {\radius*sin(\SignalPos-90)}) arc ({\SignalPos-90:\SignalPos+90:\radius});
\end{tikzpicture}
         \caption{Dithered quantization}
         \label{fig:dithered}
     \end{subfigure}
     \hfill
     \begin{subfigure}[b]{0.35\textwidth}
         \centering
\def\radius{0.4}
\def\Show{8}
\def\Offset{180/\Show}
\foreach \n in {0,...,\Show}{
\begin{tikzpicture}
\draw[thick] (0,0) circle (\radius cm);
\draw[gray, ->] (-\coor,0) -- (\coor,0);
\draw[gray, ->] (0,-\coor) -- (0,\coor);
\draw[thick] ({\radius*cos(-90+\n*\Offset)}, {\radius*sin(-90+\n*\Offset)}) circle (0.08 cm);
\draw[line width=0.5mm, red] ({\radius*cos(-90+\n*\Offset)}, {\radius*sin(-90+\n*\Offset)}) arc ({-90+\n*\Offset:-90+\n*\Offset+180:\radius});
\draw[line width=0.5mm, blue] ({\radius*cos(-90+\n*\Offset)}, {\radius*sin(-90+\n*\Offset)}) arc ({-90+\n*\Offset:-90+\n*\Offset-180:\radius});
\draw[thick] ({\radius*cos(\SignalPos)}, {\radius*sin(\SignalPos)})  node[cross,rotate=25] {};
\ifnum \numexpr \n*\Offset - 90 \relax < \SignalPos
\draw[fill] ({\radius*cos(\n*\Offset)}, {\radius*sin(\n*\Offset)}) circle (0.5mm);
\else
\draw[fill] ({\radius*cos(\n*\Offset-180)}, {\radius*sin(\n*\Offset-180)}) circle (0.5mm);
\fi
\end{tikzpicture}
\ifnum \n = 2
\\
\fi
\ifnum \n = 5
\\
\fi
}
         \caption{Dissection of dithered quantization}
         \label{fig:five over x}
     \end{subfigure}
        \caption{\small 1-bit quantizers on the unit-circle with perfect perceptual quality: ``$\times$'' indicates a sample realization of $X$; ``$\small \bullet$"  indicate the distribution of reconstruction $\hat{X}$; red and blue regions indicate the partition region associated with indices $+1$ and $-1$, respectively. In (a), the deterministic encoder is used. The sample is encoded as $+1$ and its reconstruction is distributed uniformly over the red region. In (b), the dithered approach is used, and the reconstruction would be distributed uniformly over the arc centered at the sample. There are no clear partitions in this case, and thus purple is used as a mixture of red and blue regions. In (c), ``$\circ$" indicates realizations of negative common randomness $-Z$, and the dithered quantization is viewed as a mixture of uncountably many deterministic quantizers, each associated with a realization of $Z$. }
        \label{fig:illustration}
\end{figure*}

\newcommand\M{4}
\def\Offset{180/\M}
\def\radius{1.0}
\def\coor{\radius*5/4}

\begin{figure*}[ht]
\label{fig:MultiQuan}
\centering
\vspace{-0.1cm}
\begin{tikzpicture}
\draw[->] (-\coor,0) -- (\coor,0);
\draw[->] (0,-\coor) -- (0,\coor);
\foreach \n in {0,...,\numexpr \M-1 \relax}{
 \draw[line width=0.5mm, red, semitransparent] ({\radius*cos(\n*\Offset)}, {\radius*sin(\n*\Offset)}) arc ({\n*\Offset:\n*\Offset+180:\radius});
\draw[line width=0.5mm, blue, semitransparent] ({\radius*cos(\n*\Offset)}, {\radius*sin(\n*\Offset)}) arc ({\n*\Offset:\n*\Offset-180:\radius});
}
\foreach \n in {0,...,\numexpr \M-1 \relax}{
\ifnum \numexpr \n*\Offset \relax < \SignalPos
\draw[line width=0.8mm, line cap=round, dash pattern=on 0mm off 1.9\pgflinewidth] ({\radius*cos(\n*\Offset+90-\Offset/2)}, {\radius*sin(\n*\Offset+90-\Offset/2)}) arc ({\n*\Offset+90-\Offset/2:\n*\Offset+90+\Offset/2:\radius});
\else
\draw[line width=0.8mm, line cap=round, dash pattern=on 0mm off 1.9\pgflinewidth] ({\radius*cos(\n*\Offset-90-\Offset/2)}, {\radius*sin(\n*\Offset-90-\Offset/2)}) arc ({\n*\Offset-90-\Offset/2:\n*\Offset-90+\Offset/2:\radius});
\fi
}
\draw[thick] ({\radius*cos(\SignalPos)}, {\radius*sin(\SignalPos)})  node[cross,rotate=25] {};
\node (=) at (\coor*1.3, 0) {=};
\end{tikzpicture}
\foreach \n in {0,...,\numexpr \M-1 \relax}{
\begin{tikzpicture}
\draw[->] (-\coor,0) -- (\coor,0);
\draw[->] (0,-\coor) -- (0,\coor);
 \draw[line width=0.5mm, red] ({\radius*cos(\n*\Offset)}, {\radius*sin(\n*\Offset)}) arc ({\n*\Offset:\n*\Offset+180:\radius});
\draw[line width=0.5mm, blue] ({\radius*cos(\n*\Offset)}, {\radius*sin(\n*\Offset)}) arc ({\n*\Offset:\n*\Offset-180:\radius});
\draw[thick] ({\radius*cos(\SignalPos)}, {\radius*sin(\SignalPos)})  node[cross,rotate=25] {};
\ifnum \numexpr \n*\Offset \relax < \SignalPos
\draw[line width=0.8mm, line cap=round, dash pattern=on 0mm off 1.9\pgflinewidth] ({\radius*cos(\n*\Offset+90-\Offset/2)}, {\radius*sin(\n*\Offset+90-\Offset/2)}) arc ({\n*\Offset+90-\Offset/2:\n*\Offset+90+\Offset/2:\radius});
\else
\draw[line width=0.8mm, line cap=round, dash pattern=on 0mm off 1.9\pgflinewidth] ({\radius*cos(\n*\Offset-90-\Offset/2)}, {\radius*sin(\n*\Offset-90-\Offset/2)}) arc ({\n*\Offset-90-\Offset/2:\n*\Offset-90+\Offset/2:\radius});
\fi
\end{tikzpicture}
}
\caption{\small Staggered quantizers with $1$ bit coding rate and $2$ bits common randomness.}
\label{fig:mixing}
\end{figure*}
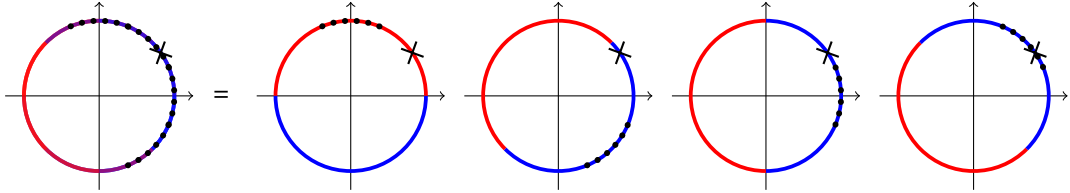
\subsection{Noise realization and staggered quantizers}

Consider again the unit circle setting at rate 1. An alternative view of a quantizer with common randomness is to consider the quantizer induced by fixing a realization of the common randomness $Z=z$, which is illustrated in Fig. \ref{fig:illustration} (c). It is seen that the partitions of these quantizers are in fact congruent to that shown in Fig. \ref{fig:illustration} (a). Since $Z$ is uniformly distributed on $[-\pi/2, \pi/2)$, the dithered quantization procedure is in fact mixing an uncountably many such quantizers, one for each $z\in [-\pi/2, \pi/2)$. Due to the common randomness $Z$, there is no need to inject decoder side randomness, which helps reduce the resultant distortion. 

The two types of quantizers considered in \cite{theis2021advantages} can then be viewed as two extremes of a class of quantizers: the former is a single quantizer with a deterministic encoder that relies solely on decoder side randomness for perception, while the latter is mixing (randomly selected using the common randomness) among uncountably many quantizers each with a deterministic encoder that requires no decoder side randomness. In between the two extremes, we can consider mixing staggered quantizers with deterministic encoders, which will need to rely on decoder side randomness to some extent. One such example with $N=4$ quantizers is illustrated in Fig. \ref{fig:mixing}. It can be seen that each individual quantizer only requires the decoder side randomness to be uniformly distributed on $1/8$ of the unit circle, instead of $1/2$ of the unit circle. As discussed earlier, decoder side randomness induces additional distortion, and this reduction in its range helps to reduce the distortion. As we increase the number of quantizers, the distortion is further reduced, eventually approaching that of the dithered quantizer.

\subsection{Staggered quantizers on the unit circle}
Generalizing the idea shown in Fig. \ref{fig:mixing}, we can use $N$ staggered $L$-level quantizers, each of which uniformly partitions the unit circle. The $N$ quantizers are obtained by offsetting sequentially by an amount of $2\pi/(LN)$ in terms of the angle on the unit circle. The common randomness uniformly selects one of $N$ quantizers, and the decoder adds private random noise uniformly distributed on $1/(2N)$ of the unit circle.

\begin{theorem} In the unit-circle setting, at perfect perceptual quality, $N$ staggered quantizers each with $L$ levels achieve the following rate-distortion pair.
\label{thm:LN}
\begin{align*}
(R,D)=\left( \log L, 2 - 2 \frac{\sin(\pi/(LN))}{\pi/(LN)} \frac{\sin(\pi/L)}{\pi/L} \right).
\end{align*}
\end{theorem}
The result subsumes the special case $N=1$ and $L=2$ given in \cite{theis2021advantages}.
\begin{proof}[Proof of Theorem \ref{thm:LN}]
Since each of $N$ quantifiers is uniform with $L$ levels, the rate for the corresponding quantization procedure is $\log L$. Due to symmetry, we analyze the distortion with a fixed quantizer. The arc (in angle) that the samples are quantized to the same index on has a length $(2\pi)/L$ since there are $L$ levels, and the inserted decoder noise is placed at the center of the arc uniformly distributed with a length $(2\pi)/(NL)$ since there are also $N$ quantizers. Since $\|(\cos(\theta), \sin(\theta))  - (\cos(\alpha), \sin(\alpha))\|^2 = 2(1 - \cos(\theta - \alpha))$, the distortion can then be calculated as
\begin{align*}
& = \frac{L}{2\pi} \frac{LN}{2\pi} \int_{-\pi/L}^{\pi/L} \left( \int_{-\pi/(NL)}^{\pi/(NL)} 2(1 - \cos(\theta - \alpha)) \dm \alpha \right) \dm \theta \\
& = 2 + \frac{L^2N}{2\pi^2} \int_{-\pi/L}^{\pi/L} \sin(\theta - \pi/(NL)) - \sin(\theta + \pi/(NL)) \dm \theta \\
& = 2 + \frac{L^2 N}{\pi^2} \left( \cos(\frac{\pi}{L}\frac{N+1}{N}) - \cos(\frac{\pi}{L} \frac{N-1}{N}) \right) \\
& = 2 - 2 \frac{\sin(\pi/(NL))}{\pi/(NL)} \frac{\sin(\pi/L)}{\pi/L},
\end{align*}
which is the desired result. 
\end{proof}

The next two theorems provide the fundamental limits of RDP coding and single-shot coding in the unit-circle setting.
\begin{theorem} In the unit-circle setting, the information-theoretic rate-distortion trade-off with perfect perceptual quality $R(D,0)$ is given by the pairs parametrized by $\lambda>0$
\label{thm:RDPunitcircle}
\begin{align*}
{\Big \{} (R,D) & = {\Big (}\log(2\pi) - h(Z), \Eb[2 - 2\cos(Z)] {\Big)} :  \\
 &\qquad\quad Z \sim p(z; \lambda) = \frac{ e^{\lambda \cos(z)} }{ \int_{-\pi}^\pi e^{\lambda \cos(z')} \dm z'},~\lambda > 0 {\Big \}}.
\end{align*}
\end{theorem}
Note that this is the best that can be achieved using infinitely large coding blocks, and it is in general impossible to achieve using single-shot coding. 

\begin{proof}[Proof of Theorem \ref{thm:RDPunitcircle}] We aim to minimize the rate-distortion Lagrangian with perfect perceptual quality for any Lagrange multiplier $\lambda > 0$, i.e, 
\begin{align}
\min_{p_{\hat{X} | X}: \hat{X} \stackrel{d}{=} X} I(X; \hat{X}) + \lambda \Eb[ \|X - \hat{X}\|^2 ]. \label{eqn:RDlambda}
\end{align}
Due to perfect perceptual quality, the reconstructed signal $\hat{X}$ must lie on the unit circle, and we can represent $\hat{X}$ by its angle $\theta(\hat{X})$. The MSE distortion term $\|X - \hat{X}\|_2^2 $ can be written as $2(1 - \cos(\theta(X) - \theta(\hat{X}))).$
The mutual information can be lower bounded by
\begin{equation}
\begin{aligned}
I(X; \hat{X}) & = h(X) - h(X | \hat{X}) \geq h(X) - h(X - \hat{X}) \\
& = h(\theta(X)) - h(\theta(X) - \theta(\hat{X})). \label{eqn:derivation}
\end{aligned}
\end{equation}
For simplicity, from here on we will write $\theta = \theta(X)$ and $\hat{\theta} = \theta(\hat{X})$, and denote $\beta:= \theta - \hat{\theta}$.

Since $h(\theta(X))$ is a constant, we can consider the optimization problem below, equivalent to lower-bounding (\ref{eqn:RDlambda})
\begin{align}
\text{minimize}_{p(\beta)} - h(\beta) + 2 \lambda \Eb[(1 - \cos(\beta) )].
\end{align}
Using simple calculus of variation, it can be verified that the optimal distribution of $\beta$ for the optimization above is $p(\beta) = \frac{ e^{2 \lambda \cos(\beta)} }{ \int_{-\pi}^\pi e^{2 \lambda \cos(\beta')} \dm \beta'}$. Since $\beta$ is independent of $\theta$, the sum $\hat{\theta} = \theta + \beta$ has a uniform distribution over $[-\pi, \pi]$. Thus this distribution indeed provides a lower bound to (\ref{eqn:RDlambda}). 

To show that they are in fact equal, we only need to observe that in (\ref{eqn:derivation}), the only inequality can be written as
\begin{align}
I(X; \hat{X})&=h(X) - h(X | \hat{X})=h(\theta) - h(\theta | \hat{\theta}) \notag\\ 
&=h(\theta) - h(\beta| \hat{\theta})\geq h(\theta) - h(\beta). \label{eqn:derivation2}
\end{align}
However, observe that we have  
\begin{align*}
p_{\beta | \hat{\theta}}(\beta | \hat{\theta}) & = \frac{ p_{\beta, \hat{\theta}}(\beta, \hat{\theta}) }{ p_{\hat{\theta}}(\hat{\theta})} =  \frac{ p_{\beta, \theta} (\beta, \hat{\theta} - \beta ) }{ p_{\hat{\theta}}(\hat{\theta})} \\
& = \frac{ p_{\beta} (\beta) p_{\theta}( \hat{\theta} - \beta ) }{ p_{\hat{\theta}}(\hat{\theta})} = p_\beta(\beta),
\end{align*}
where the last step is because both $\theta$ and $\hat{\theta}$ are uniformly distributed marginally. This implies $\beta$ is in fact independent of $\hat{\theta}$, and $h(\beta| \hat{\theta}) = h(\beta)$, and therefore (\ref{eqn:derivation2}) becomes equality, which establishes the overall equality. Thus the rate-distortion pairs are indeed characterized by that given in Theorem \ref{thm:RDPunitcircle}.

It is not difficult to verify that the curve (or function) above is continuous, and its epigraph is non-empty and closed lying in the upper right quadrant. Each point on the curve naturally has a supporting hyperplane, since it is a solution of optimizing the corresponding Lagrangian. Thus by the partial converse of supporting hyperplane theorem the curve is convex. 
\end{proof}

\begin{theorem} In the unit-circle setting, the optimal scalar quantization (single shot coding) trade-off between the coding rate and the distortion with perfect perceptual quality is the piece-wise linear function with the following extreme points
\label{thm:OptimalQuan}
\begin{align*}
\left\{(R,D) = \left( \log L, 2 - 2\frac{\sin(\pi/L)}{\pi/L} \right): L = 1,2,3,\ldots \right\},
\end{align*}
which can be achieved by dithered quantizations. 
\end{theorem}

As $N\rightarrow \infty$, we see that $\frac{\sin(\pi/(LN))}{\pi/(LN)}\rightarrow 1$, therefore, the performance of the staggered quantizer approaches that of dithered quantization in this setting. Due to the uniform data source distribution, dithered quantizers are optimal, and $N$ staggered quantizers each with $L$ levels each does not offer any advantage over dithered quantizers. However, as we will discuss in the next section, this is not the case in general, since the flexibility in entropy coding can lead to an additional edge. 
\begin{proof}[Proof of Theorem \ref{thm:OptimalQuan}]
Any codecs $(f, g)$ can be represented by $f: \Xc \times \Rb \rightarrow \Zb$ and $g: \Zb \times \Rb \rightarrow \Xc$. The signal $X$ is encoded by $f(X, V)$ to some integer and then reconstructed by $\hat{X} = g( f(X, V), V )$, where $V$ is the common randomness. 

Due to the perfect perceptual quality requirement, the reconstructed signal $\hat{X}$ must lie on the unit circle. Without considering perceptual quality, we first characterize the scalar optimal quantization under the condition that reconstruction $\hat{X}$ lies on the unit circle. Take any Lagrange multiplier $\lambda > 0$, consider minimizing the following rate distortion Lagrangian with decision variables $(f, g, V)$
\begin{align*}
& H(f(X; V) | V) + \lambda \Eb_{X, V}[ d(X,  g( f(X ; V) ; V) ) ] \notag \\
& = \Eb_{V}[ \Eb_{X}[ -\log(\Pb(f(X; V) | V)) + \lambda d(X,  g( f(X ; V) ; V) ) | V] ]
\end{align*}
It suffices to study the deterministic quantizer, since for any stochastic quantizer $(f, g, V)$, there exists a deterministic quantizer $(f(; v), g(;, v))$ with some realization of $V = v$ such that its Lagrangian is at most that of the stochastic quantizer. 

It is straightforward to verify that the optimal deterministic quantizer in this setting must have contiguous regions (pathological cases may exist for complex distributions \cite{gyorgy2002structure}), i.e., the region in $\Xc$ of the same index $f(\cdot,v)$ should be contiguous.  For such a quantizer with $L$ levels, i.e., $|f(\cdot,v)|=L$, it can then be shown using calculus that it must be a uniform quantizer. The optimal scalar quantization (single shot coding) trade-off between the coding rate and the distortion is the piece-wise linear function with the following extreme points
\begin{align*}
\left\{(R,D) = \left( \log L, 2 - 2\frac{\sin(\pi/L)}{\pi/L} \right): L = 1,2,3,\ldots \right\}.
\end{align*}
This piece-wise linear function is a lower bound, when considering perfect perceptual quality. However, it is straightforward to verify that dithered quantization has the perfect perceptual quality and can achieve the extreme points and thus match the lower bound. Thus the optimal scalar quantization trade-off between the coding rate and the distortion with perfect perceptual quality is also the piece-wise linear function above and can be achieved by time-sharing dithered quantizers.
\end{proof}


\section{Design of staggered quantizers for general scalar sources} \label{sec:meth}

Consider applying the staggered quantization approach to a general scalar source. Assuming there are $N$ uniform quantizers to be staggered, the encoding function $f_n(x)$ for the $n$-th quantizer with stepsize $\Delta$ is 
\begin{align}
f_n(x)=\left\lfloor\frac{x}{\Delta}-\frac{n}{N}\right\rceil, \quad n=0,1,2,\ldots,N-1, \label{eqn:uniform-quantizer}
\end{align}
where $\lfloor\cdot\rceil$ is the operation that rounds to the nearest integer.

To achieve perfect perceptual quality, decoder side randomness must be used, yet due to the potential non-uniformity of the distribution, it is more involved than simply subtracting certain random values. To present the procedure, first denote the density of the data source $X$ as $p_X(x)$ and denote by $F_X(x) = \Pb(X \leq x )$ its cumulative distribution function. Denote its inverse as $F^{-1}_X(t) \triangleq \inf\{x:  F_X(x) > t\}$ for any $t \in [0, 1)$. Let us introduce a density function on $[a,b]$ as $q_{a,b}(x) \triangleq \frac{p_X(x)}{\int_{a}^b p_X(t) dt}$. A random variable generated privately at the decoder side according to this distribution is denoted as $\tilde{Z}_{a,b}$, which is independent of all the other random variables.

Define an indexing function $m(x, n) = N \cdot f_n(x) + n$, which essentially specifies an order of all the quantization cells in all these $N$ quantizers. Define its inverse at input $x$ as $m^{-1}_X(j)\triangleq \inf\{x: \exists n \in [0:N-1], m(x, n) = j\}$. Intuitively, for each quantizer and quantizer cell index pair $(n,f_n(x))$, the reconstruction at the decoder is a random variable that follows a distribution that matches the data sample distribution in an interval. Now to specify the specific interval, we 
define a sequence of boundaries $(a(j), b(j))_{j \in \Zb}$ as 
\begin{align*}
a(j) \triangleq F_X^{-1}\left( \sum_{k =  1}^{N} \frac{F_X(m_x^{-1}(j - k))}{N} \right),~ b(j-1) \triangleq a(j). 
\end{align*}

The encoding and reconstruction process can now be described as follows. Given data source $X$ at the encoder side, the encoding procedure uniformly at random selects one of the $N$ encoders $\{f_0, f_1, \ldots, f_{N-1}\}$ with stepsize $\Delta$. The index $n$ of the selected encoder is a common randomness shared by the decoder, and the data sample is encoded as $f_n(X)$. At the decoder, we compute the index $j$ using $f_n(x)$ and $n$ by the indexing function $m(\cdot)$, and the reconstruction is a random sample $\hat{X} = \tilde{Z}_{a(j), b(j)}$. More formally, the decoding function upon receiving code $f_n(X) = i$ is
\begin{align}
g(i) = \tilde{Z}_{a(j), b(j)},~\text{ with } j = Ni + n, \label{eqn:MutliQuan-generator}
\end{align}
where $n$ is the common randomness of the offset quantizer index. We remark here that the offsets can be viewed as a random dither which takes discrete values in $\{0,1/N,2/N, \ldots,(N-1)/N\}$. However, for each realization, the reconstruction is sampled in an interval, unlike in classic deterministic quantizers or dithered quantizers. An illustration is given in Fig. \ref{fig:general}.

\begin{figure}[t]
\begin{center}
 \includegraphics[width=0.47\textwidth]{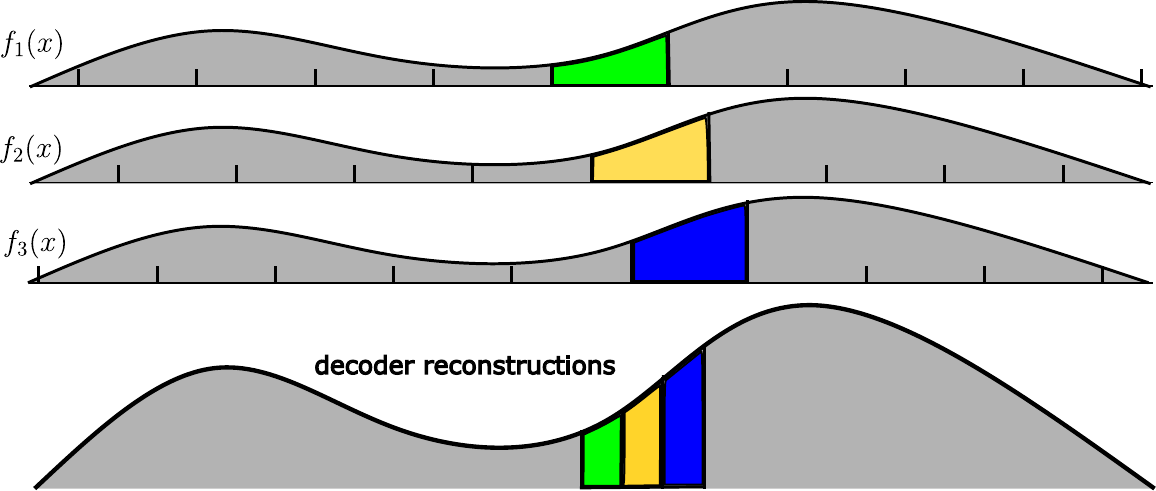}
 \caption{\small Staggered quantizers for general probability distributions.}
  \label{fig:general}
\end{center}
\end{figure}

Since the number of staggered quantizers is small, it is possible to design tailored entropy code for each, whereas this is impossible for dithered quantizers, resulting in a rate close to $H(f(X+Z))$. Dithered quantization also suffers because $X+\tilde{Z}$ induces loss of perception, and an additional shaping step is required. As shown in Fig. \ref{fig:uniform}, the proposed approach can sometimes outperform both dithered quantizers and deterministic encoders. Particularly, even mixing $2$ quantizers appears to provide competitive performance. 


\begin{figure}[t]
\begin{center}
 \includegraphics[width=0.47\textwidth]{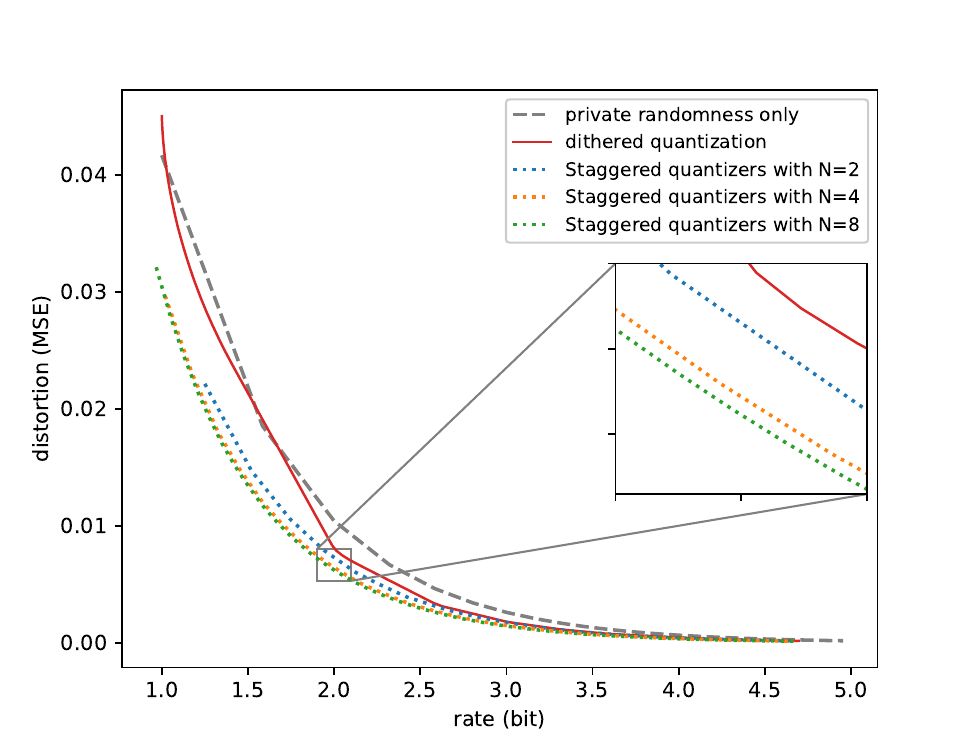}
 \caption{\small Quantization of a uniformly distributed source on an interval}
  \label{fig:uniform}
\end{center}
\end{figure}

\section{Conclusion}
We consider RDP coding from a quantizer design perspective. By decomposing dithered quantization, we obtain staggered quantizers as intermediates between the two extremes of dithered quantization and quantization without common randomness. This new perspective provides a new way to understand one-shot coding for RDP. 

\bibliographystyle{IEEEtran}

\appendix

\begin{proof}[Optimality of Uniform Quantizers in the Proof of Theorem \ref{thm:OptimalQuan}]
Consider two adjacent Voronoi cells. Suppose the two adjacent regions have a total size (in terms of the angle spanned) $2\pi r$ for some $r \in (0, 1]$, moreover, suppose the first Voronoi is of size $2\pi \alpha$ for some $\alpha \in (0, r)$. For optimal partitions, $\alpha$ must be a minimizer of the following function
\begin{align*}
l(\alpha; r) &= (r - \alpha) \ln(r - \alpha) + \frac{\lambda}{\pi} \sin(\pi (r - \alpha)) \notag\\
&+ \alpha \ln(\alpha) + \frac{\lambda}{\pi} \sin(\pi \alpha).
\end{align*}
Its derivative is
\begin{align*}
l'(\alpha; r) =& - \ln(r - \alpha) - \lambda \cos(\pi(r - \alpha)) \notag\\
&\qquad+ \ln(\alpha) + \lambda \cos(\pi \alpha)
\end{align*}
and its second derivative is
\begin{align*}
l''(\alpha; r) = \frac{1}{r - \alpha} + \frac{1}{\alpha} -  \lambda \pi ( \sin(\pi(r - \alpha)) + \sin(\pi \alpha) ).
\end{align*}
It is not hard to verify that $l$ and $l''$ are even functions, and $l'$ is an odd function. There are two circumstances
\begin{enumerate}
\item $\lambda$ is small, and $l''(\alpha; r) \geq 0$. Then $l(\alpha; r)$ is a non-constant symmetric convex function whose optimal value is achieved by $\alpha \rightarrow 0$ or $\alpha \rightarrow r$, which conflicts with the fact that the optimal quantizer has non-empty Voronoi. 
\item $\lambda$ is large, and $l''(\alpha; r)$ will be positive on both ends and negative in the middle. $l'(\alpha; r)$ is increasing, decreasing and increasing. $l(\alpha; r)$ will either have a maximum with $\alpha = r/2$ or the maximum is approached by  $\alpha \rightarrow 0$ or $\alpha \rightarrow r$.
\end{enumerate}
Therefore any two adjacent non-empty Voronoi cells have the same size. The optimal quantizer thus must have equal-sized Voronoi cells, thus a uniform quantizer. 
\end{proof}

\end{document}